\documentclass[runningheads]{llncs}
\usepackage[utf8]{inputenc}

\usepackage{amsfonts}
\usepackage{amsmath}
\usepackage{xspace}
\usepackage{xcolor}
\usepackage{tabularx}

\usepackage{dutchcal} 

\usepackage{graphicx}

\usepackage{caption}
\usepackage{subcaption}
\usepackage{comment} 
\usepackage{mathtools}

\usepackage[mathscr]{euscript}

\usepackage{hyperref}
\usepackage{cleveref}

\usepackage[normalem]{ulem}

\usepackage{yhmath} 

\usepackage[misc]{ifsym} 

\newcommand\IR[1]{$\dag$\footnote{IR: #1}}

\newcommand{\ourpath}{\ensuremath{\pi}\xspace} 

\newcommand{\pset}{\ensuremath{\Pi}\xspace} 

\newcommand{\union}{~\ensuremath{{\large \cup}\kern-0.5em\raisebox{0.75ex}{\tiny o}}~\xspace}

\newcommand{\intersec}{\xspace\ensuremath{{\large \cap}\kern-0.5em\raisebox{0.5ex}{\tiny o}}~\xspace}

\newcommand{\compl}{\xspace\ensuremath{{\large \setminus}\kern-0.2em\raisebox{0.5ex}{\tiny o}}~}

\newcommand{\ourwatertanksafeprop}{\ensuremath{\psi_{wt}}\xspace}

\newcommand{\tankcount}{\ensuremath{J}\xspace}

\newcommand{\tanksize}{\ensuremath{TS}\xspace}

\newcommand{\wl}{\ensuremath{w}\xspace}

\newcommand{\wlp}{\ensuremath{\hat{w}}\xspace}

\newcommand{\inflow}{\ensuremath{in}\xspace}

\newcommand{\outflow}{\ensuremath{out}\xspace}

\newcommand{\uct}{\ensuremath{UT}\xspace}

\newcommand{\lct}{\ensuremath{LT}\xspace}

\newcommand{\mdl}{\ensuremath{\mathsf{M}}\xspace}

\newcommand{\mcp}{\ensuremath{\mathsf{M}_{\text{OS}}}\xspace}

\newcommand{\mba}{\ensuremath{\mathsf{M}_{\text{AS}}}\xspace}

\allowdisplaybreaks

\begin{document}

\title{Conservative Safety Monitors of\\ Stochastic Dynamical Systems
}
%
%

\author{Matthew Cleaveland\inst{1}$^{\text{(\Letter)}}$ \and Oleg Sokolsky\inst{1} \and Insup Lee\inst{1}  \and Ivan Ruchkin\inst{2}}

\authorrunning{M. Cleaveland et al.}

\institute{University of Pennsylvania, Philadelphia, PA, USA 19104 \\ \email{\{mcleav,sokolsky,lee\}@seas.upenn.edu} \and University of Florida, Gainesville, FL, USA 32611 \\ \email{iruchkin@ece.ufl.edu} }
\maketitle              
\begin{abstract}

Generating accurate runtime safety estimates for autonomous systems is vital to ensuring their continued proliferation. However, exhaustive reasoning about future behaviors is generally too complex to do at runtime. To provide scalable and formal safety estimates, we propose a method for leveraging design-time model checking results at runtime. Specifically, we model the system as a probabilistic automaton (PA) and compute bounded-time reachability probabilities over the states of the PA at design time. At runtime, we combine distributions of state estimates with the model checking results to produce a bounded time safety estimate. We argue that our approach produces well-calibrated safety probabilities, assuming the estimated state distributions are well-calibrated. We evaluate our approach on simulated water tanks.


\end{abstract}


\section{Introduction}

\looseness=-1
As autonomous systems see increased use and perform critical tasks in an open world, reasoning about their safety and performance is critical. In particular, it is vital to know if a system is likely to reach an unsafe state in the near future.

\looseness=-1
The field of predictive runtime monitoring offers ways for performing this reasoning. The basic idea is to reason about the expected future behaviors of the system and its properties. However, accurately computing future system states is computationally infeasible at runtime, as it requires running expensive reachability analysis on complex models. Previous works have computed libraries of reachability analysis results at design time and used them at runtime \cite{chou2020predictive}. But these approaches require the system dynamics to have certain invariances to reduce the number of times reachability analysis must be called offline. 

\looseness=-1
Other lines of work use system execution data to learn discrete probabilistic models of the system, which are then used to perform predictive runtime monitoring, as there is rich literature for runtime monitoring of discrete automata. These models range from discrete-time Markov chains (DTMCs) \cite{Babaee2019} to hidden Markov models (HMMs) \cite{Babaee2018} to Bayesian networks \cite{Jaeger2020}. However, it is difficult to provide guarantees relating the performance of the automata models to the real system, due to the fact that they are fit using finite data. Of particular interest is ensuring the models are conservative: it is essential to avoid run-time overconfidence in the safety of the dynamical system.

\looseness=-1
In this paper, we propose a method for predictive run-time monitoring of safety probabilities that builds on the strengths of the existing works. We use a mix of conservative modeling techniques and data-driven modeling techniques to transform the dynamical system into a probabilistic automaton (PA).\footnote{In our scope, PAs are equivalent to Markov Decision Processes (MDPs) without rewards: both have finite states with  probabilistic and non-deterministic transitions.} We then employ probabilistic model checking (PMC) to compute the safety of the model over all its states offline. Finally, we synthesize lightweight monitors that rely on the model checking results and a well-calibrated state estimator to compute the probability of system safety at runtime.

\looseness=-1
Under the assumption that the PA model is conservative and that the state estimator is well-calibrated, we prove that our runtime monitors are conservative. We demonstrate that our modeling technique is likely to result in conservative PA models. 
Finally, we show that our method produces well-calibrated, accurate, and conservative monitors on a case study using water tanks.

The contributions of this paper are threefold:
\begin{itemize}
    \item We present a method for conservatively modeling dynamical systems as PAs and using PMC results at runtime to monitor the system's safety.
    \item We prove that if our PA models are conservative then the monitor safety estimates will be conservative.
    \item We demonstrate our approach on a case study of water tanks. We empirically show that our PA models and runtime monitors are both conservative.
\end{itemize}
\looseness=-1
The rest of the paper is structured as follows. We give an overview of the related work in \Cref{sec:related}, provide the necessary formal background in \Cref{sec:background}, and formulate the problem in \Cref{sec:problem}. \Cref{sec:approach} goes over our proposed approach and \Cref{sec:guarantees} provides formal conservatism guarantees for the approach. We describe the results of our case study in \Cref{sec:caseStudy} and conclude in \Cref{sec:conclusion}.

\section{Related Work}
\label{sec:related}

We divide the previous works in the area of predictive runtime monitoring into two bins. The first bin analyzes dynamical system models, while the second analyzes automata models.


\subsubsection{Dynamical systems approaches }
\looseness=-1
A large portion of the predictive monitoring for dynamical systems literature focuses on reasoning about the safety of autonomous vehicles. Prior work has employed reachability analysis to estimate the future positions of other cars to estimate the safety of a proposed vehicle trajectory \cite{althoff2009}. In \cite{jasour2022fast}, the authors develop techniques to estimate the probability of a proposed trajectory resulting in a collision with other vehicles, which are given as distributions of states predicted by neural networks (NNs). In \cite{chou2020predictive}, the authors use precomputed reachability analysis and Bayesian inference to compute the probability of an autonomous vehicle colliding with static obstacles. This approach requires the system dynamics to have certain invariances to ensure the reachability analysis can be feasibly run at design time. This approach is conceptually similar to ours, but we employ automata-based abstractions instead of making invariance assumptions about the system dynamics. 

\looseness=-1
Previous works have also addressed the problem of synthesizing runtime monitors for signal temporal logic (STL) properties of dynamical systems. Approaches range from conformal prediction \cite{cairoli2022conformal,Lindemann2022ConformalPF}, design time forward reachability analysis \cite{yu2022model}, computing safe envelopes of control commands \cite{Yoon2019}, online linear regression \cite{Granig2020}, and uncertainty aware STL monitoring \cite{ma2021predictive}.

\subsubsection{Automata approaches} The first works of this type developed predictive LTL semantics, also called LTL$_3$ \cite{zhang2012runtime,leucker2012sliding}, for discrete automata. The LTL$_3$ semantics allowed to the system to determine if every infinite extension of an observed finite trace would satisfy or not satisfy a specification. Recent work has extended these ideas to timed systems \cite{Pinisetty2017}, multi-model systems \cite{Ferrando2022}, and systems with assumptions \cite{Cimatti2021}. Another approach uses neural networks to classify if unsafe states of a hybrid automaton (HA) can be reached from the current state of the HA \cite{Bortolussi2019,Bortolussi2021,Francesca2021}. They additionally use conformal prediction to get guarantees about the accuracy of their predictions \cite{shafer2008tutorial}. However, these frameworks give very coarse predictions, as they can only determine if a system is guaranteed to be safe, guaranteed to be unsafe, or not sure.

\looseness=-1
Another thread of work uses data to learn probabilistic models that can then be used in conjunction with predictive monitoring techniques. In \cite{Babaee2018}, the authors learn an HMM model of the system from simulation data and perform bounded reachability analysis to determine the probability of an LTL specification being violated from each state of the HMM. This work was extended using abstraction techniques to simplify the learned models \cite{Babaee2018_2}. In \cite{Babaee2019}, the same authors employ importance sampling to efficiently learn discrete-time Markov chain (DTMC) models from data, which they then use to synthesize predictive monitors. In \cite{Jaeger2020}, the authors use Bayesian networks to model temporal properties of stochastic timed automata. The Bayesian networks are updated online to improve their performance. Finally, in \cite{Ferrando2021} the authors use process mining techniques to learn predictive models of systems, which are in turn used to synthesize predictive runtime monitors. 
An interesting line of future work for us is exploring applying our runtime monitoring technique using these models as they are updated from new observations online. 

\looseness=-1
The most similar work to ours presents two methods for synthesizing predictive monitors for partially observable Markov decision processes (POMDPs)~\cite{junges2021runtime}. The first approach combines precomputed safety probabilities of each state with POMDP state estimators to estimate the probability that the system will remain safe. However, state estimation of POMDPs is computationally expensive since the set of potential state distributions increases exponentially due to the non-determinism in the model. The second approach uses model checking of conditional probabilities to directly compute the safety of the system based on the observation trace. A downside of this approach is that it requires running model checking at runtime. Our method, on the other hand, avoids expensive computations at run time while maintaining design-time scalability through abstraction. 

\section{Background}
\label{sec:background}

In the following \Cref{def:pa,def:parcomp,def:prob}, borrowed from Kwiatkowska et al. \cite{Kwiatkowska2013}, we use $Dist(S)$ to refer to the set of probability distributions over a set $S$, $\eta_s$ as the distribution with all its probability mass on $s \in S$, and $\mu_1 \times \mu_2$ to be the product distribution of $\mu_1$ and $\mu_2$.
\begin{definition} \label{def:pa}
A \emph{probabilistic automaton} 
(PA) is a tuple $\mdl=(S,\bar{s},\alpha,\delta,L)$, where $S$ is a finite set of states, $\bar{s}\in S$ is the initial state, $\alpha$ is an alphabet of action labels, $\delta \subseteq S \times \alpha \times Dist(S)$ is a probabilistic transition relation, and $L: S \rightarrow 2^{AP}$ is a labeling function from states to sets of atomic propositions from the set AP. 
\end{definition} 
If $(s,a,\mu)\in \delta$ then the PA can make a transition in state $s$ with action label $a$ and move based on distribution $\mu$ to state $s'$ with probability $\mu(s')$, which is denoted by $s \xrightarrow{a} \mu$. If $(s,a,\eta_{s'})\in \delta$ then we say the PA can transition from state $s$ to state $s'$ via action $a$. A state $s$ is terminal if no elements of $\delta$ contain $s$. A path in $M$ is a finite/infinite sequence of transitions $\ourpath = s_0 \xrightarrow{a_0,\mu_{0}} s_1 \xrightarrow{a_1, \mu_{1}}\hdots$ with $s_0=\bar{s}$ and $\mu_{i}(s_{i+1})>0$. A set of paths is denoted as $\pset$. We use $\mdl(s)$ to denote the PA $\mdl$ with initial state $s$.

Reasoning about PAs also requires the notion of a \textit{scheduler}, which resolves the non-determinism during an execution of a PA. For our purposes, a scheduler $\sigma$ maps each state of the PA to an available action label in that state. We use $\pset_{\mdl}^{\sigma}$ for the set of all paths through $\mdl$ when controlled by scheduler $\sigma$ and $Sch_{\mdl}$ for the set of all schedulers for $\mdl$. Finally, given a scheduler $\sigma$, we define a probability space $Pr_{\mdl}^{\sigma}$ over the set of paths $\pset_{\mdl}^{\sigma}$ in the standard manner.

Given PAs $\mdl_1$ and $\mdl_2$, we define parallel composition as follows:

\begin{definition}\label{def:parcomp}
    The \emph{parallel composition}    of PAs $\mdl_1=(S_1,\bar{s}_1,\alpha_1,\delta_1,L_1)$ and $\mdl_2=(S_2,\bar{s}_2,\alpha_2,\delta_2,L_2)$ is given by the PA $\mdl_1~||~\mdl_2=(S_1 \times S_2, (\bar{s}_1,\bar{s}_2),\alpha_1 \cup \alpha_2, \delta,L)$, where $L(s_1,s_2)=L_1(s_1)\cup L_2(s_2)$ and $\delta$ is such that $(s_1,s_2) \xrightarrow{a} \mu_1 \times \mu_2$ iff one of the following holds: (i) $s_1\xrightarrow{a}\mu_1,s_2\xrightarrow{a}\mu_2$ and $a \in \alpha_1 \cap \alpha_2$, (ii) $s_1 \xrightarrow{a}\mu_1, \mu_2=\eta_{s_2}$ and $a \in (\alpha_1 \setminus \alpha_2)$, (iii) $\mu_1=\eta_{s_1},s_2\xrightarrow{a}\mu_2$ and $a \in (\alpha_2 \setminus \alpha_1)$.
\end{definition}


In this paper, we are concerned with probabilities that the system will not enter an unsafe state within a bounded amount of time. These are represented as bounded-time safety properties, which we express using metric temporal logic (MTL)~\cite{koymans1990specifying}.
Following the notation from \cite{Katoen2013}, we denote these properties as $$\square^{\leq T} s \notin S_{unsafe},$$
where $S_{unsafe} \subset S$ is the set of unsafe states and $T\geq 0$ is the time bound.

\begin{definition}\label{def:prob}
For MTL formula $
\psi$, PA $\mdl$, and scheduler $\sigma \in Sch_{\mdl}$, the \emph{probability of $\psi$ holding} is:
\begin{align*}
    Pr_{\mdl}^{\sigma}(\psi) \coloneqq Pr_{\mdl}^{\sigma}\{ \pi \in \pset_{\mdl}^{\sigma}~|~\pi \models \psi \},
\end{align*}
where $\pi \models \psi$ indicates that the path $\pi$ satisfies $\psi$ in the standard MTL semantics~\cite{koymans1990specifying}. We specifically consider MTL safety properties, which are MTL specifications that can be falsified by a finite trace though a model. 
\end{definition}

Probabilistically verifying an MTL formula $\psi$ against $M$ requires checking that the probability of satisfying $\psi$ meets a probability bound for all schedulers. This involves computing the minimum or maximum probability of satisfying $\psi$ over all schedulers:
\begin{align*}
    Pr_{\mdl}^{min}(\psi) &\coloneqq \operatorname{inf}_{\sigma\in Sch_{\mdl}} Pr_{\mdl}^{\sigma}(\psi) \\ Pr_{\mdl}^{max}(\psi) &\coloneqq \operatorname{sup}_{\sigma\in Sch_{\mdl}} Pr_{\mdl}^{\sigma}(\psi)
\end{align*}

We call $\sigma$ a min scheduler of \mdl if $Pr_{\mdl}^{\sigma}(\psi) = Pr_{\mdl}^{min}(\psi)$. We use $Sch_{\mdl}^{min}$ to denote the set of min schedulers of \mdl.

\emph{Remark:} For the rest of this paper, we use $Pr$ when referring to model-checking probabilities and $P$ for all other probabilities.

\subsubsection{Calibration and Conservatism}

Consider a scenario where a probability estimator is predicting probability $\hat{p}$ that a (desirable) event $E$ will occur (e.g., a safe outcome). We define the calibration for the probability estimates (adapted from Equation (1) of \cite{guo2017calibration}):

\begin{definition}[Calibration]\label{def:genericCalibration}
    The probability estimates $\hat{p}$ of  event $E$ are \emph{well-calibrated} if 
    \begin{align}\label{eq:genericCalibration}
        P(E ~|~ \hat{p}=p)=p, \quad \forall p \in [0,1]
    \end{align}
\end{definition}

Next, we define conservatism for the probability estimates:

\begin{definition}[Conservative Probability]\label{def:genericConservative}
    The probability estimates $\hat{p}$ of a desirable event $E$ are \emph{conservative} if 
    \begin{align}\label{eq:genericConservative}
        P(E ~|~ \hat{p}=p)\geq p, \quad \forall p \in [0,1]
    \end{align}
\end{definition}

In other words, the estimates $\hat{p}$ are conservative if they underestimate the true probability of event $E$. Note that any monitor that is well-calibrated (\Cref{def:genericCalibration}) is guaranteed to be conservative (\Cref{def:genericConservative}), but not vice versa.

Two standard metrics for assessing the calibration of the $\hat{p}$ estimates are \emph{expected calibration error} (ECE) \cite{guo2017calibration} and \emph{Brier score} \cite{ranjan_combining_2010}. The ECE metric is computed by dividing the $\hat{p}$ values into equally spaced bins in $[0,1]$, within each bin taking the absolute difference between the average $\hat{p}$ and the empirical probability of event $E$, and weighted-averaging across bins with their sizes as weights. So ECE penalizes discrepancies between the estimator confidence and empirical probability of $E$ within each bin. The Brier score is the mean squared error of the probability estimates $$ \sum_i (\hat{p_i}-\mathbf{1}_{E_i})^2 $$

\section{Problem Statement}
\label{sec:problem}

Consider the following discrete-time stochastic system titled \mcp with dynamics:
\begin{align}
\begin{split}
    X(t+1)&=f(X(t),U(t))),\\
    Y(t)&=g(X(t),V(t)), \\
    \bar{X}(t),\bar{Z}(t)&=h(\bar{Z}(t-1),Y(t),W(t)),\\
    U(t)&=c(\bar{X}(t)),
\end{split}
\label{eq:system}
\end{align}
where $X(t) \in S \subset \mathbb{R}^n$ is the system state (with bounded $S$); $Y(t) \in \mathbb{R}^p$ are the observations; $\bar{X}(t) \in \mathbb{R}^n$ is the estimated state of the system; $\bar{Z}(t)\in \mathbb{R}^z$ is the internal state of the state estimator (e.g., a belief prior in a Bayesian filter); $U(t)\in\mathcal{U}\subset\mathbb{R}^m$ is the control output, which we discretize, resulting in a finite number $|\mathcal{U}|$ of control actions, 
 the functions $f:\mathbb{R}^n\times\mathbb{R}^m\to\mathbb{R}^n$, $g:\mathbb{R}^n\times\mathbb{R}^v\to\mathbb{R}^p$, $h:\mathbb{R}^z \times \mathbb{R}^p\times \mathbb{R}^w \to \mathbb{R}^n\times\mathbb{R}^z$ describe the system dynamics, perception map, and state estimator respectively; the function $c:\mathbb{R}^p\to\mathbb{R}^m$ is a stateless controller; and $V(t)\in D_v\subseteq\mathbb{R}^v$ and $W(t)\in D_w\subseteq\mathbb{R}^w$  describe perception and state estimator noise. The $V(t)$ noise models inexact perception, such as an object detector missing an obstacle. The $W(T)$ noise accounts for state estimators that use randomness under the hood. A common example of this is particle filters randomly perturbing their particles so that they do not collapse to the exact same value.

Let $S_{unsafe} \subset S$ denote the set of unsafe states of \mcp. At time $t$, we are interested in whether \mcp will lie in $S_{unsafe}$ at some point in the next $T$ time steps. This is represented by the bounded time reachability property
\begin{align}\label{eq:McpProperty}
    \psi_{\mcp} = \square^{\leq T} \left(X \notin S_{unsafe} \right)
\end{align}

Let $P(\psi_{\mcp}~|~\bar{Z}(t))$ denote the probability of \mcp satisfying $\psi_{\mcp}$. Our goal is to compute calibrated (\Cref{def:genericCalibration}) and conservative (\Cref{def:genericConservative}) estimates of $P(\psi_{\mcp}~|~\bar{Z}(t))$ 
at runtime, which we denote as $\widehat{P}\left(\psi_{\mcp}~|~\bar{Z}(t)\right)$.

\section{Overall Approach}
\label{sec:approach}

Our approach consists of a design time and runtime portion. At design time, a PA of the system (including its dynamics, perception, state estimation, and controller) is constructed using standard conservative abstraction techniques. Then the bounded-time safety probability for each state of the model is computed using model checking and stored in a look-up table. At runtime, the estimated state (or distribution of states) from the real system's state estimator is used to estimate the abstract state (or distribution of abstract states) of the abstract system. This abstract state (or distribution of states) is used in conjunction with the lookup table to estimate the bounded-time safety of the real system.

\subsection{Design Time}
\label{sec:approachDesignTime}

The design time aspect of our approach has two parts. First, we convert the original system \mcp into a probabilistic automaton \mba.
Then we use probabilistic model checking to compute the bounded time safety of \mba for each state in the model.

\subsubsection{Model Construction}

\looseness=-1
To convert \mcp into a probabilistic automaton, \mba, we first need to create probabilistic models of the perception $g$ and state estimation $h$ components of \mcp. To do this, we simulate \mcp and record the perception errors $X(t)-\bar{X}(t)$. We discretize the domain of these errors and estimate a categorical distribution over it. For example, this distribution would contain information such as ``the perception will output a value that is between 2m/s and 3m/s greater than the true velocity of the car with probability $1/7$.''

\looseness=-1
To convert the system dynamics $f$ and controller $c$ to a probabilistic automaton, we use a standard interval abstraction technique. The high-level idea is to divide the state space $S$ of \mcp into a finite set of equally sized hyperrectangles, denoted as $S'$. So every $s_1' \in S'$ has a corresponding region $S_1 \subset S$. \mba then has a transition from $s_1'$ to $s_2'$ (in \mba) if at least one state in $S_1$ has a transition to a state in $S_2$ (in \mcp) under some control command $u\in\mathcal{U}$. Note that state $s_1'$ can non-deterministically transition to multiple states in $S'$ because it covers an entire hyperrectangle of states in \mcp. This ensures that the interval abstraction is conservative, as it overapproximates the behaviors of \mcp.

\looseness=-1
Finally, the perception error model, controller, and interval abstraction are all parallel-composed into a single model as per \Cref{def:parcomp}.

\looseness=-1
\textit{Remark: } In describing the construction of the \mba, we have not mentioned anything about initial states: we do not keep track of a singular initial state for \mba. Instead, we will later run model checking for the full range of initial states of \mba to anticipate all runtime scenarios. For our purposes, the ``initial state-action space'' of \mba consists of every abstract state and control action. We include the control action in the initial state space because when using the model's safety probabilities online, we know what the next control action will be.


\subsubsection{Safety Property}
We need to transform the bounded time safety property on \mcp given in \Cref{eq:McpProperty} into an equivalent property on \mba. To do this, we compute the corresponding set of unsafe states on \mba, which is defined as $$ S'_{unsafe} = \{ s' ~|~ \exists s\in S_{unsafe}, s' \text{ corresponds to } s  \}$$

Letting $s'$ denote the state of \mba, the bounded time safety property for \mba is 
\begin{align}
    \psi_{\mba} \coloneqq \square^{\leq T} \left( s' \notin S'_{unsafe} \right) 
\end{align}


\subsubsection{Probabilistic Model Checking}



\looseness=-1
The final design-time step of our approach computes the safety probability of \mba for every state in the model. 
This step amounts to computing the below values using standard model checking tools:
$$Pr_{\mba(s',u)}^{min}(\psi_{\mba}), \;\; \forall s' \in S', \; \forall u \in \mathcal{U}$$

\looseness=-1
This requires running model checking on \mba for a range of initial states, which can be a time-consuming process. To mitigate this, we note that \mba is simpler to analyze than \mcp, since the size of the state space gets reduced during the interval abstraction process. Additionally, one can lower the time bound $T$ on the safety property to further speed up the model checking. 

\looseness=-1
The probabilities from the model checking are stored in a lookup table, which we denote as $G(s',u)$. It will be used at runtime to estimate the likelihood of the system being unsafe in the near future.

\looseness=-1
\emph{Remark:} This approach would work for any bounded time MTL properties, 
however more complex formulas may take longer to model check.

\subsection{Runtime}
\label{sec:approachRuntime}

At runtime, we observe the outputs of the state estimator and controller and run them through the lookup table to compute the probability of the system avoiding unsafe states for the next $T$ time steps. We propose two different ways of utilizing the state estimator. The first way is to simply use the point estimate from the state estimator. In cases of probabilistic estimators, this means taking the mean of the distribution. The second way uses the estimated state distribution from the state estimator. This requires an estimator with a probabilistic output, but most common state estimators, such as particle filters and Bayesian filters, keep track of the distribution of the state. The second way takes full advantage of the available state uncertainty to predict safety.

\subsubsection{Point Estimate}
At time $t$, the state estimator outputs state estimate $\bar{X}(t)$. The controller then outputs control command $U(t) = c(\bar{X}(t))$. Finally, we get a safety estimate $\hat{P}^{mon}_{point}(\bar{X}(t),U(t))$ by plugging $\bar{X}(t)$ and $U(t)$ into $G$:
\begin{align}\label{eq:stateMonitor}
    \hat{P}^{mon}_{point}(\bar{X}(t),U(t))) = G(\bar{X}(t),U(t))
\end{align}

\subsubsection{State Distribution}
Now assume that at time $t$ state estimator additionally outputs a state estimate $\bar{X}(t)$ and a finite, discrete distribution of the state, denoted as $P_{\bar{X}(t)}$. 
The controller still outputs control command $U(t)=c(\bar{X}(t))$. To estimate the safety of the system, we compute a weighted sum of the safety of each state in $P_{\bar{X}(t)}$ using $G$ and $U(t)$:
\begin{align}\label{eq:distMonitor}
    \hat{P}^{mon}_{dist}(P_{\bar{X}(t)},U(t)) = \sum_{s\in Supp\left(P_{\bar{X}(t)}\right)} P_{\bar{X}(t)}(s) \cdot G(s',U(t))
\end{align}
where $Supp\left(P_{\bar{X}(t)}\right)$ denotes the (finite) support of $P_{\bar{X}(t)}$, $P_{\bar{X}(t)}(s)$ denotes the estimated probability of \mcp being in state $s$ according to $P_{\bar{X}(t)}$, and $s' \in S'$ is the state in \mba that corresponds to state $s\in S$ in \mcp.

\section{Conservatism Guarantees}
\label{sec:guarantees}

This section proves that our state-distribution monitoring produces safety estimates that are \emph{conservative} and \emph{well-calibrated}; that is, we underestimate the probability of safety. We require two assumptions for that. The first assumption is the conservatism of abstract model \mba, by which we mean that its probability of being safe is always less than that of \mcp for the same initial condition. The second assumption is the calibration of the state estimator, which means that it produces state probabilities that align with the frequencies of these states. Below we formalize and discuss these assumptions before proceeding to our proof. 


\begin{definition}[Model Conservatism]\label{def:conservative}
    Abstraction \mba is \emph{conservative} with respect to system \mcp if
        \begin{align}
    P_{\mcp(s,u)}(\psi) \geq Pr_{\mba(s',u)}^{min}(\psi) \; \; \forall s \in S, u \in \mathcal{U}
\end{align}
where $s' \in S'$ is the state in \mba that corresponds to state $s$ in \mcp.

\end{definition}

In general, it is difficult to achieve provable conservatism of \mba by construction: the model parameters of complex components (e.g., vision-based perception) are estimated from data, and they may have complicated interactions with the safety chance. Instead, we explain why our approach is likely to be conservative in practice and validate this assumption in the next section. 

\looseness=-1
Consider \mcp and \mba as compositions of two sub-models: dynamics/control and perception/state estimation. We construct \mba such that its dynamics/control component always overapproximates the dynamics/control portion of \mcp. That means that any feasible sequence of states and control actions from \mcp is also feasible in \mba. This follows from the use of reachability analysis over the intervals of states to compute the transitions of \mba.


\looseness=-1
It is unclear how to formally compare the conservatism of perception/state estimation portions of \mba and \mcp when they are created from simulations of the perception/state estimation component of \mcp. First, these components are not modeled explicitly due to the high dimensionality of learning-based perception. Thus, when estimating probabilities from samples, we essentially approximate the average-case behavior of the component. Second, it is often unknown in which direction the probabilities need to be shifted to induce a conservative shift to the model. One opportunity is to use monotonic safety rules~\cite{cleaveland2022}; for now, this remains a promising and important future research direction.

\looseness=-1
To summarize, the dynamics/control portion of \mba overapproximates that of \mcp, while the perception/state estimation portion of \mba approximates the average-case behavior of \mcp. So one would expect, on average, \mba to be conservative with respect to \mcp, even though we cannot formally prove that.

\looseness=-1
Next, we define the calibration for the state estimator (adapted from Equation (1) of \cite{guo2017calibration}):

\begin{definition}[Calibration]\label{def:calibration}
    Given the dynamical system from \Cref{eq:system} and state estimator $h$ that outputs a discrete, finite distribution of the estimated state, denoted $P_{\bar{x}(t)}$, we say that $h$ is \emph{well-calibrated} if 
    \begin{align}\label{eq:estimatorCalibration}
        P(x(t)=s ~|~ P_{\bar{x}(t)}(s)=p)=p, \quad \forall p \in [0,1]
    \end{align}
\end{definition}

\looseness=-1
Intuitively, what this definition means is that if the state estimator says that there is probability $p$ that the system is in state $s$, then the system will be in state $s$ with probability $p$. Calibration is an increasingly common requirement for learning-based detectors~\cite{guo2017calibration,minderer_revisiting_2021,gong_confidence_2021,ruchkin_confidence_2022} and we validate it in our experiments.

\looseness=-1
Now we are ready for our main theoretical result: assuming that \mba is conservative with respect to \mcp and that the state estimator is well-calibrated, we show that the safety estimates of our monitoring are conservatively calibrated. 


\begin{theorem}
    Let the system \mcp in \Cref{eq:system} be given with a well-calibrated state estimator (\Cref{def:calibration}). Let \mba be a conservative model 
    of \mcp (\Cref{def:conservative}). Finally, assume that the safety of \mcp conditioned on the true state of the system is independent of the safety estimate from the monitor. Given state estimator distribution $P_{\bar{X}(t)}$ and control command $U \in \mathcal{U}$, the safety estimates from the state distribution monitor (\Cref{eq:distMonitor}) are conservative:
    \begin{align}\label{eq:conservativeCalibration}
        P( \psi_{\mcp}~|~\hat{P}^{mon}_{dist}(P_{\bar{X}(t)},U(t))=p) \geq p \quad \forall p\in[0,1]
    \end{align}
\end{theorem}
\begin{proof}
    We start with conditioning the safety of the system on the state of the system and proceed with equivalent transformations:
    \begin{align*}
        P( \psi_{\mcp}~|~\hat{P}^{mon}_{dist}(P_{\bar{X}(t)},U(t))=p) &= \\
        \int_{s \in S} P\Big( \psi_{\mcp}~|~X(t)=s,\hat{P}^{mon}_{dist}\big(P_{\bar{X}(t)},U(t)\big)=p\Big)\cdot \quad & \\
        P\Big(X(t)=s~|~\hat{P}^{mon}_{dist}\big(P_{\bar{X}(t)},U(t)\big)=p\Big) ds &= \\
        \int_{s \in S} P( \psi_{\mcp}~|~X(t)=s) \cdot P_{\bar{X}(t)}(s) ds &= \\
        \sum_{s \in P_{\bar{X}(t)}} P( \psi_{\mcp}~|~X(t)=s) \cdot P_{\bar{X}(t)}(s) &= \\
        \sum_{s \in P_{\bar{X}(t)}} P_{\mcp(s,U(t))}(\psi)\cdot P_{\bar{X}(t)}(s) &\geq \\
        \sum_{s \in P_{\bar{X}(t)}} Pr_{\mba(s_{\downarrow},U(t))}^{min}(\psi)\cdot P_{\bar{X}(t)}(s) &= p
    \end{align*}
\end{proof}

\looseness=-1
The first step comes from marginalizing the state $X(t)$ into the left side of \Cref{eq:conservativeCalibration}. The second step comes from the assumption that the safety of the system given the state is independent of the monitor output and the assumed calibration of the monitor from \Cref{eq:estimatorCalibration}. The third step follows from the discrete, finite support of the state estimator output and the calibration. The fourth step comes from substituting and rearranging terms. The final step comes from the assumed conservatism of \mba in \Cref{def:conservative}.

\section{Case Study}
\label{sec:caseStudy}

Our experimental evaluation aims to demonstrate that the safety estimates from our monitoring approach are conservative and accurate. Additionally, we compare the effect of using the point-wise and distribution-wise state estimation. We perform the evaluation on a simulated water tank system and use the PRISM model checker~\cite{Kwiatkowska2011PRISM} to perform the probabilistic model checking. The code and models for the experiments can be found \href{https://github.com/earnedkibbles58/CPSMonitoringNFM2023/tree/master}{on Github}.

\subsection{Water Tanks}

Consider a system consisting of \tankcount water tanks, each of size \tanksize, draining over time, and a central controller that aims to maintain some water level in each tank. With $\wl_i[t]$ as the water level in the $i^{th}$ tank at time $t$, the discrete-time dynamics for the water level in the tank is given by:
\begin{equation}
    \wl_i[t+1] = \wl_i[t] - \outflow_i[t] +  \inflow_i[t],
\end{equation}
where $\inflow_i[t]$ and $\outflow_i[t]$ are the amounts of water entering (``inflow'') and leaving (``outflow'') respectively the $i^{th}$ tank at time $t$. The inflow is determined by the controller and the outflow is a constant determined by the environment.

Each tank is equipped with a noisy sensor to report its current perceived water level, \wlp, which is a noisy function of the true current water level, \wl. The noise on the sensor outputs is a Gaussian with zero mean and known variance. Additionally, with constant probability the perception outputs $\wlp=0$ or $\wlp=\tanksize$.

\looseness=-1
Each water tank uses a standard Bayesian filter as a state estimator. The filter maintains a discrete distribution over the system state. On each perception reading, the filter updates its state distribution using a standard application of the Bayes rule. The mean of the state distribution at this point is the estimator's point prediction, which is sent to the controller. Once the control action is computed, the filter updates its state distribution by applying the system dynamics.

The central controller has a single source of water to fill one tank at a time (or none at all) based on the estimated water levels. Then this tank receives a constant value $\inflow>0$ of water, whereas the other tanks receive 0 water. Each tank has a local controller that requests itself to be filled when its water level drops below the lower threshold \lct and stops requesting to be filled after its water level reaches the upper threshold \uct. If several tanks request to be filled, the controller fills the one with the lowest water level (or, if equal, it flips a coin).

At runtime, we are interested in the probability that a tank will neither be empty or overflowing, represented by the bounded-time safety property: $$ \ourwatertanksafeprop \coloneqq \square^{\leq 10} \vee_{i = 1..\tankcount} \left( wl_i > 0 \wedge  wl_i < \tanksize \right) $$

\looseness=-1
\subsubsection{Model Construction}
We construct the \mba model for $\tankcount=2$ water tanks, $\inflow=13.5$, $\outflow_i[t]=4.3$, $\tanksize=100$, $\lct=10$, $\uct=90$, and water level intervals of size 1 by following the description in \Cref{sec:approachDesignTime}. To model the combination of perception and state estimation, we estimated the state distributions with $100$ trials of $50$ time steps. \Cref{fig:waterTankPerceptionErrors} shows a histogram of the state estimation errors.


\begin{figure}
    \centering
    \includegraphics[width=0.6\linewidth]{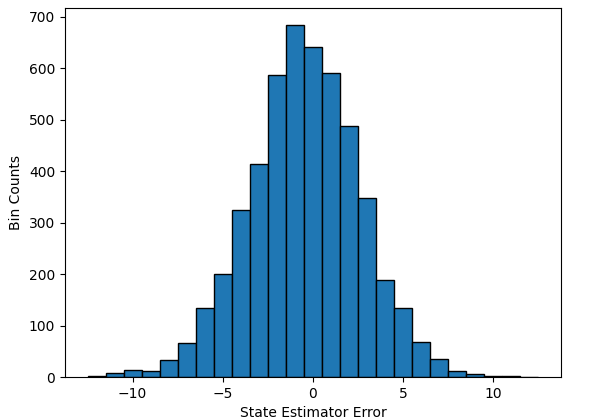}
    \caption{Histogram of state estimation errors for the water tanks.}
    \label{fig:waterTankPerceptionErrors}
    \vspace*{-3mm}
\end{figure}

\subsubsection{Model Checking}
The initial state of \mba comprises the water level of each tank, the low-level control command of each tank, and the filling command of the central controller. There are $101$ discrete water levels in each tank and $5$ possible configurations of the 3 control commands, for a total of $51005$ different initial states of \mba. We model-checked \ourwatertanksafeprop in these initial states on a server with 80 Intel(R) Xeon(R) Gold 6148 CPU @ 2.40GHz CPUs by running 50 parallel PRISM instances at a time. The full verification process took approximately 24 hours, which is acceptable for the design-time phase.

\subsection{Results}

To test our approach, we ran $500$ trials of the water tanks starting from water levels between $40$ and $60$. Each trial lasted for $50$ time steps (recall that the model checking checked $10$ time steps into the future) and $74$ trials resulted in a water tank either over- or underflowing. We evaluated three different monitors in our approach. One used the point estimates from the Bayesian filter (``point estimate monitor''), another used the estimated distribution from the state estimator (``state distribution monitor''), and the last used the true state of the system (``true state monitor'', for comparison only).


\subsubsection{Qualitative performance} \Cref{fig:waterTankTrials} shows the safety estimates of the monitors for one safe and one unsafe trial. The monitors keep high safety estimates for the entirety of the safe trial. In the unsafe trial, the failure occurred at time step 42 due to a tank overflowing. The safety estimates are high at first but then begin to drop around time step 30, predicting the failure with a 10-step time horizon.

\begin{figure}
\centering
\begin{subfigure}{0.45\textwidth}
  \centering
  \includegraphics[width=1\linewidth]{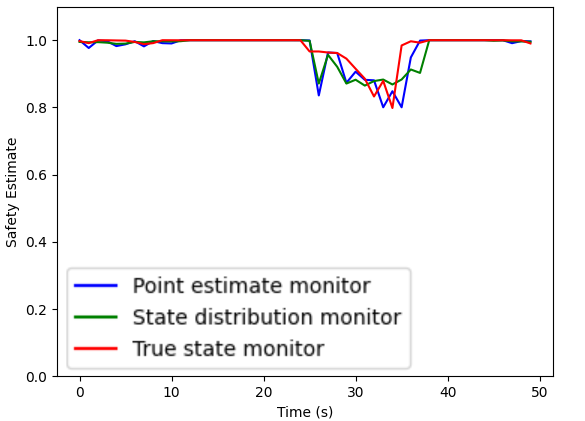}
  \label{fig:waterTankSafeTrial}
  \caption{\small Safe trial}
\end{subfigure}%
\begin{subfigure}{0.45\textwidth}
  \centering
  \includegraphics[width=1\linewidth]{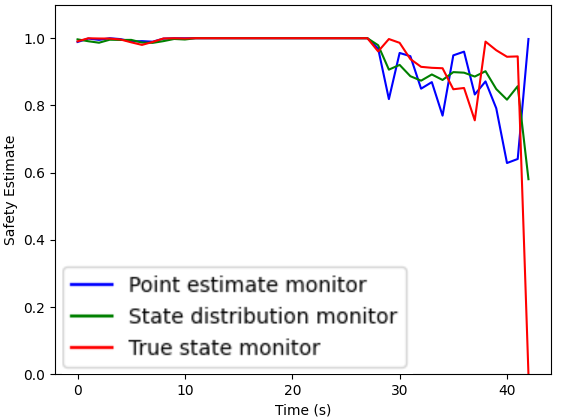}
  \label{fig:waterTankUnsafeTrial}
  \caption{Unsafe trial}
\end{subfigure}
\caption{\small Monitor safety estimates for two water tank trials.}
\label{fig:waterTankTrials}
\vspace*{-3mm}
\end{figure}

\subsubsection{Calibration} Next, to examine the overall calibration of our safety estimates, we bin the safety estimates into 10 bins of width 0.1 ($[0-0.1,0.1-0.2,\hdots,0.9-1]$) and compute the empirical safety chance within each bin. The results are shown in \Cref{fig:monitorCalibration}, with the caveat that we only plot bins with at least 50 samples to ensure statistical significance. The point estimate monitor and true state monitor are conservative for all of their bins. On the other hand, the state distribution monitor has the best overall calibration. We also computed the ECE and Brier scores for the monitors, which are shown in \Cref{tab:monitorMetrics}. To assess the conservatism of the monitors, we introduce a novel metric called \emph{expected conservative calibration error} (ECCE). It is similar to ECE, except that it only sums the bins where the average monitor confidence is greater than the empirical safety probability (i.e., the cases where the monitor is overconfident in safety). The ECCE values for the monitors are also shown in \Cref{tab:monitorMetrics}. Note that $ECE \geq ECCE$, because ECCE only aggregates a subset of the bins that ECE does. Our results show that the monitors are well-calibrated and conservative, and that the state distribution monitor manages to capture the uncertainty particularly well. 

\begin{figure}
    \centering
    \begin{subfigure}{0.32\textwidth}
        \includegraphics[width=1\linewidth]{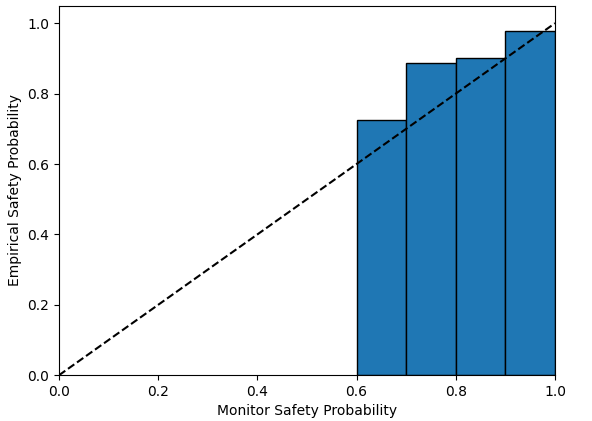}
        \label{fig:calibrationStateMonitor}
        \caption{Point estimate monitor \newline}
    \end{subfigure}
    \begin{subfigure}{0.32\textwidth}
        \includegraphics[width=1\linewidth]{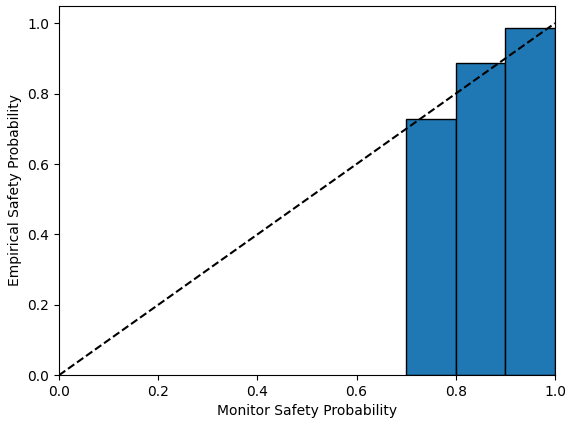}
        \label{fig:calibrationDistMonitor}
        \caption{State distribution monitor}
    \end{subfigure}
    \begin{subfigure}{0.32\textwidth}
        \includegraphics[width=1\linewidth]{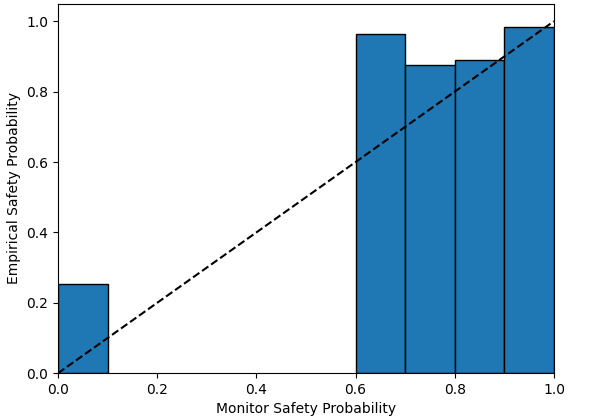}
        \label{fig:calibrationTrueStateMonitor}
        \caption{True state monitor \newline}
    \end{subfigure}
    \caption{Calibration plots for the three monitors. The x-axis shows the binned safety estimates reported by the monitor and the y-axis shows the empirical safety probability. The diagonal dashed line denotes perfect calibration. Bars higher than the dashed line represent under-confidence (i.e., conservatism) and bars lower than the dashed line represent over-confidence. }
    \label{fig:monitorCalibration}
    \vspace*{-3mm}
\end{figure}


\subsubsection{Accuracy} Finally, we are interested in the ability of the monitors to distinguish safe and unsafe scenarios. To do this, we computed a receiver operating characteristic (ROC) curve for the three monitors, shown in \Cref{fig:waterTankROC}, and areas under curve (AUC) in \Cref{tab:monitorMetrics}. As expected, the state distribution monitor and true state monitor outperform the point estimate monitor. One surprising aspect is that the state distribution monitor performs about as well as the true state monitor. We hypothesize that this is because the state distribution contains information about how well the state estimator will perform in the near future. Investigating this potential phenomenon is another area of future work.

\begin{figure}
    \centering
    \includegraphics[width=0.6\linewidth]{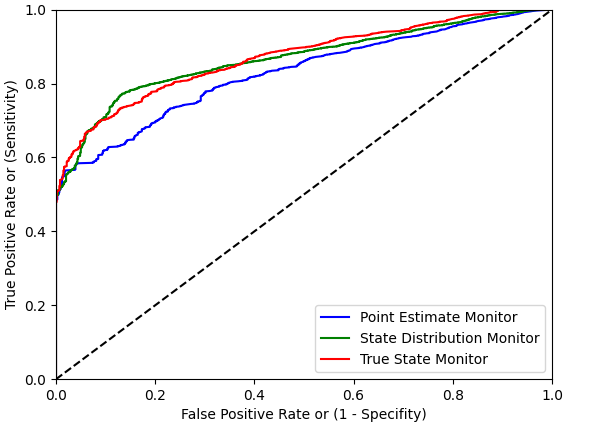}
    \caption{ ROC curves for the three monitors.}
    \label{fig:waterTankROC}
    \vspace*{-3mm}
\end{figure}

\begin{table}[]
    \centering
    \vspace*{-3mm}
    \caption{Calibration and classification metrics for the monitors. }
    \begin{tabular}{|c|c|c|c|c|}
        \hline Monitor Type & ECE & ECCE & Brier Score & AUC  \\
        \hline State estimate &  $0.0157$ & $0.00818$ & $0.0292$ & $0.828$ \\
        \hline State distribution &  $0.00252$ & $0.000899$ & $0.0275$ & $0.867$  \\
        \hline True state &  $0.0129$ & $0.00459$ & $0.0273$ & $0.870$ \\
        \hline
    \end{tabular}
    \label{tab:monitorMetrics}
    \vspace*{-3mm}
\end{table}

\subsubsection{Validation of assumptions}
First, we empirically validate whether \mba is conservative with respect to \mcp. Directly verifying this claim is infeasible, since it requires computing $P_{\mcp(s,u)}(\psi)$ for an infinite number of states $s \in S$. However, we can examine the performance of the true state monitor as a proxy for the conservatism of \mba: the true state monitor obtains the probabilities from \mba using the true state, avoiding any sensing and state estimation noise. The slightly underconfident true state monitor bins in \Cref{fig:monitorCalibration} and the very low ECCE in \Cref{tab:monitorMetrics} both provide strong evidence that \mba is indeed conservative.

\looseness=-1
Second, we examine the calibration assumption of the state estimator. We computed its ECE across all water levels, resulting in the negligible value of $0.00656$. We conclude that this state estimator gives calibrated results in practice.

\section{Conclusion}

\label{sec:conclusion}

This paper introduced a method for synthesizing conservative and well-calibrated predictive runtime monitors for stochastic dynamical systems. Our method abstracts the system as a PA and uses PMC to verify the safety of the states of the PA. At runtime, these safety values are used to estimate the true safety of the system. We proved that our safety estimates are conservative provided the PA abstraction is conservative and the system's state estimator is well-calibrated. We demonstrated our approach on a case study with water tanks. Future work includes applying our method to existing approaches that learn discrete abstractions directly from data, exploring how to construct conservative perception/state estimation abstractions, using our prior work in \cite{cleaveland2022} to reduce the number of model checking calls, and investigating the effects of the estimated state distribution's variance on the future system safety.


\section*{Acknowledgments}

This work was supported in part by DARPA/AFRL FA8750-18-C-0090 and by ARO W911NF-20-1-0080. Any opinions, findings and conclusions or recommendations expressed in this material are those of the authors and do not necessarily reflect the views of the Air Force Research Laboratory (AFRL), the Army Research Office (ARO), the Defense Advanced Research Projects Agency (DARPA), or the Department of Defense, or the United States Government.


\bibliographystyle{splncs04}
\bibliography{literature}

\end{document}